\documentclass[12pt]{article}

\usepackage{hyperref}

\usepackage{amsfonts, amsmath, amssymb, amsthm}
\usepackage{a4}

\usepackage{upgreek}

\DeclareMathOperator{\const}{const}
\DeclareMathOperator{\Ker}{Ker}
\DeclareMathOperator{\LinearImage}{Im}
\newcommand{\tildetimes}{\mathbin{\widetilde{\times}}}
\newcommand{\strutik}{\vrule height 2.6ex depth 1ex width 0pt}

\newtheorem{theorem}{Theorem}
\newtheorem{lemma}{Lemma}

\theoremstyle{remark}
\newtheorem*{remark}{Remark}

\author{I.G. Korepanov, N.M. Sadykov}
\title{Hexagon cohomologies and polynomial TQFT actions}
\date{July 2017--January 2018}

\begin{document}

\sloppy

\maketitle

\begin{abstract}
Hexagon relations are combinatorial or algebraic realizations of four-dimen\-sional Pachner moves. We introduce some simple set-theoretic hexagon relations and then `quantize' them using what we call `polynomial hexagon cohomologies'. Based on this, topological quantum field theories are proposed with polynomial `discrete Lagrangian densities' taking values in finite fields. First calculations of the resulting manifold invariants, arising from polynomial cocycles of degree three and in characteristic two, show their nontriviality.
\end{abstract}

\section{Introduction}\label{s:i}

Let $M$ be a triangulated closed $n$-dimensional PL (\,=\,piecewise linear) manifold. Suppose we can `color' the $m$-faces of this triangulation, for a fixed $m<n$, that is, assign an element of a fixed set~$X$ called the \emph{set of colors}, to each $m$-face (in this paper, we will consider the case $n=4$, \ $m=3$). For each of the simplices~$u$ of the highest dimension~$n$, let there be given a set
\[
R_u \subset \underbrace{X\times \dots \times X}_{\substack{\text{one copy of }X\\ \text{for each }m\text{-face of }u}} 
\]
of `permitted colorings' of (all $m$-faces of)~$u$. A coloring of (all $m$-faces in) a simplicial complex is called permitted if its restrictions on all the $n$-simplices are permitted.

In particular, we can consider the set~$R_M$ of permitted colorings for the whole~$M$. Suppose it is finite, and we want to build somehow a PL manifold invariant based on its cardinality~$|R_M|$. We know that any triangulation of~$M$ can be transformed into any other triangulation using a sequence of \emph{Pachner moves}~\cite{Pachner,Lickorish}. A Pachner move transforms a cluster of $n$-simplices into another cluster having the same boundary and occupying the same place in a triangulation. We will sometimes call the initial cluster the \emph{l.h.s.}\ of the move, while the final cluster its \emph{r.h.s.}

A natural desire is thus to have some nice correspondence between the permitted colorings for the l.h.s.\ and r.h.s.\ of any Pachner move. This idea leads to what can be called \emph{combinatorial realizations of Pachner moves}. These can also be called `set-theoretic $(n+2)$-gon relations', because such terminology appears to be commonly associated with Pachner moves in $n$ dimensions, see, e.g.,~\cite{DM}.

In the present paper, we work in dimension $n=4$. Accordingly, our first subject of interest will be set-theoretic \emph{hexagon} (\,=\,6-gon). We will see that it leads to invariants of the form
\begin{equation}\label{R/cN}
\frac{|R_M|}{\const_0^{N_0}\cdot\const_4^{N_4}},
\end{equation}
where $N_k$ is the total number of $k$-simplices in the triangulation.

Invariant~\eqref{R/cN}, or at least its numerator---the total number of colorings---resembles, in our opinion, \emph{quandle invariants} in the theory of knots and their higher analogues, see, e.g.,~\cite{CKS}. The next fruitful idea known from that theory consists in using \emph{quandle cohomology}. It turns out that similar constructions work in our case as well, and  there exist some \emph{very nontrivial} `hexagon cocycles'. As a result, the number~$|R_M|$ in~\eqref{R/cN} splits into a sum, where each summand separately produces an invariant.

Introducing cocycles can be likened, in our opinion, to a quantization of a `classical' theory. Note also that our invariants can be rewritten in a form that deserves the name of `discrete analogue of Feynman integral', with permitted colorings playing the role of `fields', and corresponding cocycle values playing the role of `Lagrangian density'.

Below,
\begin{itemize}\itemsep 0pt
 \item in Section~\ref{s:set-theoretic}, we consider set-theoretic hexagon relations,
 \item in Section~\ref{s:l}, we provide an explicit example of such relation using linear spaces over a finite field,
 \item in Section~\ref{s:cohomology}, we introduce our polynomial hexagon cohomology,
 \item in Section~\ref{s:actions}, we write out explicitly all the nontrivial polynomial cocycles for polynomial degrees $\le 6$ and characteristics $\le 5$. These cocycles can be interpreted as discrete Lagrangian densities for TQFT's (\,=\,topological quantum field theories),
 \item in Section~\ref{s:results}, we present first calculation results for manifold invariants arising from one of these TQFT's,
 \item finally, in Section~\ref{s:d}, we discuss directions for further research.
 \item Also, there is Appendix~\ref{s:app}, where we briefly explain some results appearing in~\cite{bosonic} and relevant for the present paper.
\end{itemize}

\subsection*{Division of labor between the authors}
The discovery of polynomial TQFT's in finite characteristics is due to I.K., starting with one cubic polynomial~\eqref{2,3a} in characteristic two, see conference report~\cite{report}. Further development of the theory presented here is a joint work. Computer calculations of manifold invariants were done by N.S.\ using, in particular, his specialized GAP \href{https://sourceforge.net/projects/plgap/}{package} for calculations with PL manifolds.

\section{Four-dimen\-sional Pachner moves and set theoretic hexagon relations}\label{s:set-theoretic}

\subsection{Pachner moves}\label{ss:P}

Consider a 5-simplex $\Delta^5=123456$ (i.e., whose vertices are numbered from 1 to~6). Its boundary~$\partial \Delta^5$ consists of six pentachora (\,=\,4-simplices). Imagine that $k$ of these pentachora, $1\le k\le 5$, enter in a triangulation of PL manifold~$M$. Then we can replace them with the remaining $6-k$ pentachora, without changing~$M$. This is called Pachner move, and there are five kinds of them: 1--5, 2--4, 3--3, 4--2, and 5--1; here the number before the dash is~$k$, while the second number is, of course, $6-k$.
 
\subsection{Ordering of vertices}\label{ss:o}

We assume that all vertices in every triangulation we are considering are \emph{ordered}, most often---due to the fact that they are \emph{numbered} from~1 to their total number. Of course our manifold invariants introduced in Sections \ref{s:l} and~\ref{s:cohomology} will be independent of this ordering.

When we speak about an individual pentachoron, like in formula~\eqref{R12345} below, we may denote it as~$12345$---but we keep in mind that our constructions or/and statements are also valid for any pentachoron~$ijklm$, \ $i<j<k<l<m$, if we do the obvious replacements $1\mapsto 1$, \ldots, $5\mapsto m$.

Similarly, the description of Pachner moves in the above Subsection~\ref{ss:P} stays valid, of course, for any vertices $i,\ldots,n$ instead of~$1,\ldots,6$.

\subsection{Permitted colorings}\label{ss:perm}

Let a finite set~$X$ be given, called \emph{the set of colors}. We will assign a color $\mathsf x\in X$ to each \emph{3-face} of a simplicial complex such as~$\Delta^5$ or triangulated manifold~$M$. Not all colorings, however, are \emph{permitted}.

For one pentachoron, let it be 12345, permitted colorings are determined by definition by a given subset
\begin{equation}\label{R12345}
R_{12345}\subset X_{2345}\times X_{1345}\times X_{1245}\times X_{1235}\times X_{1234}
\end{equation}
in the Cartesian product of copies of~$X$ corresponding to its 3-faces. For other pentachora~$u$ in question, permitted colorings are determined by copies~$R_u$ of~$R_{12345}$, in accordance with our Subsection~\ref{ss:o}.

\begin{remark}
An interesting theory appears also in the case of \emph{nonconstant}~$R_u$, i.e., if the permitted colorings of different pentachora are different. In this paper, we will, however, confine ourselves to the case of constant~$R_u$.
\end{remark}

For a cluster~$C$ of pentachora obtained by gluing them along their 3-faces---such as the l.h.s.\ or r.h.s.\ of a Pachner move, or the whole manifold~$M$---permitted coloring is by definition such a coloring of all 3-faces (including inner faces where gluing has been done) whose restriction onto each pentachoron is permitted. We denote the set of permitted colorings of~$C$ as $R_C\subset \prod_{t\subset C} X_t$.

\subsection{Full set theoretic hexagon}\label{ss:full}

Take any subcomplex~$C\subset\partial \Delta^5$ containing $k$ pentachora, \ $1\le k\le 5$, and its complementary subcomplex~$\bar C$ containing the remaining $6-k$ pentachora. We impose the following requirements on the sets~$R_u$:

\begin{enumerate}\itemsep 0pt
 \item\label{i:CbarC} the restrictions onto $\partial C=\partial \bar C$ of permitted colorings of any such~$C$ make up the same set of colorings of this common boundary as the restrictions of permitted colorings of~$\bar C$,
 \item\label{i:ak} additionally, there are fixed natural numbers---multiplicities~$a_k$, \ $1\le k\le 5$, such that to any chosen coloring of $\partial C=\partial \bar C$ correspond exactly $a_k$ colorings of~$C$ (that is, taking into account its inner 3-faces), and exactly $a_{6-k}$ colorings of~$\bar C$.
\end{enumerate}

Item~\ref{i:ak} says thus that any coloring of a boundary extends in a fixed number of ways onto the whole $C$ or~$\bar C$.

\medskip

If both conditions \ref{i:CbarC} and~\ref{i:ak} hold, we will say that the \emph{full set theoretic hexagon}, or simply \emph{full hexagon}, is satisfied.

\section{Linear set theoretic hexagon}\label{s:l}

\subsection{Linear set of colors}\label{ss:ll}

Now let the set~$X$ of colors be a \emph{two-dimensional linear space} over a \emph{finite} field~$F$. We identify this space with~$F^2$, and write its elements as
\begin{equation}\label{sft}
\mathsf x_t = \begin{pmatrix} x_t \\ y_t \end{pmatrix} \in X=F^2,
\end{equation}
if we mean the color of tetrahedron~$t$. For each pentachoron~$u$, let
\begin{equation}\label{Ru}
R_u\subset \bigoplus_{\text{tetrahedra\;}t\subset u}X_t
\end{equation}
be a \emph{five-dimensional linear subspace} in the ten-dimensional direct sum of copies of~$X$.

Na\"ively speaking, we leave free to change exactly half of the ten parameters corresponding to the five 3-faces of~$u$.
    
\subsection{A constant full hexagon solution}

Specifically, we introduce, and will be using throughout this paper, the following set~$R_u$ of permitted colorings. For pentachoron $u=12345$, subspace $R_u$~\eqref{Ru} is, by definition, given by the following five linear relations:
\begin{equation}\label{y-any-char}
\begin{pmatrix} y_{2345}\\ y_{1345}\\ y_{1245}\\ y_{1235}\\ y_{1234} \end{pmatrix} = 
\begin{pmatrix}0 & -2 & 1 & 1 & -2\\
0 & -1 & 0 & 1 & -1\\
-1 & 2 & -2 & 0 & 1\\
-1 & 3 & -2 & -1 & 2\\
0 & 1 & -1 & 0 & 0\end{pmatrix}
\begin{pmatrix} x_{2345}\\ x_{1345}\\ x_{1245}\\ x_{1235}\\ x_{1234} \end{pmatrix}. 
\end{equation}
The same definition holds also for other pentachora, with only the relevant substitution of indices, see Subsection~\ref{ss:o}.

Below in this Subsection we prove some statements, relying mostly on direct calculations using the explicit form of the $5\times 5$ matrix in~\eqref{y-any-char}, as well as \eqref{d:psi} for the mapping~$\psi$ introduced below. The reader interested in the \emph{origin} of these explicit formulas can find a brief explanation in Appendix~\ref{s:app}.

We will also denote the columns of $x$'s and~$y$'s in~\eqref{y-any-char} as $\mathbf x$ and~$\mathbf y$, respectively, while the $5\times 5$ matrix as~$\mathcal R$. Equation~\eqref{y-any-char} acquires then the form
\begin{equation}\label{R}
\mathbf y=\mathcal R\mathbf x.
\end{equation}

A fundamental property of matrix~$\mathcal R$ is related to an $F$-linear mapping~$\psi$ that we are going to define,
\begin{equation}\label{psi}
\psi\colon\quad (F\text{-colorings of edges}) \to (\text{colorings of }M).
\end{equation}
Here $M$ is a triangulated 4-manifold, maybe with boundary, and with ordered vertices (as we agreed in Subsection~\ref{ss:o}). $F$-colorings of edges are formal linear combinations of edges with coefficients in~$F$; the edges are understood as \emph{unoriented}: $ij=ji$. Colorings of~$M$ are, according to Subsections \ref{ss:perm} and~\ref{ss:ll}, $F^2$-colorings of tetrahedra---formal linear combinations of tetrahedra with coefficients in~$F^2\ni \begin{pmatrix} x \\ y \end{pmatrix}$; the tetrahedra are of course also understood as unoriented. Below, when we introduce an edge $b=ij$ or a tetrahedron $t=ijkl$, we assume that the vertices go in the increasing order: $i<j$ or, respectively, $i<j<k<l$.

It will be convenient for us to regard the linear mapping~$\psi$ as a matrix whose rows and columns correspond to the tetrahedra and edges, respectively, and whose entries are \emph{2-columns}:
\begin{equation}\label{xy}
\psi_{t,b} = \begin{pmatrix} x \\ y \end{pmatrix}.
\end{equation}
These columns are defined as follows. First,
\begin{equation}\label{tnb}
\text{if}\quad t\not\supset b,\quad \text{then} \quad x=y=0. 
\end{equation}
Second, for a tetrahedron $t=ijkl$ there are six edges $b\subset t$, and we write the six corresponding columns together as a $2\times 6$ matrix as follows:
\begin{equation}\label{d:psi}
\begin{pmatrix} \psi_{t,ij} & \psi_{t,ik} & \psi_{t,il} & \psi_{t,jk} & \psi_{t,jl} & \psi_{t,kl} \end{pmatrix} = \begin{pmatrix} -1 & 2 & -1 & -1 & 0 & 1 \\ 1 & -1 & 0 & 0 & 1 & -1 \end{pmatrix} .
\end{equation}

\begin{theorem}\label{th:d:R}
The image of any $F$-coloring of edges under mapping~$\psi$ gives a\/ \emph{permitted} coloring:
\begin{equation}\label{d:su}
\LinearImage \psi \subset R_M .
\end{equation}
\end{theorem}

\begin{proof}
It is clearly enough to prove this for just one pentachoron, $M=u=12345$, and its ten edges $b=12,\ldots,45$. This can be done by a direct calculation.
\end{proof}

\begin{lemma}\label{l:d:R}
For $M=\partial \Delta^5$  (recall that this is the l.h.s.\ and r.h.s.\ of any Pachner move together, see Subsection~\ref{ss:P}), image of~$\psi$ gives \emph{all} permitted colorings.
\end{lemma}

\begin{proof}
Indeed, a calculation shows that, in this case,
\[
\dim\LinearImage \psi=\dim R_{\partial\Delta^5} =9.
\]
\end{proof}

\begin{theorem}\label{th:g}
Full hexagon holds indeed if permitted colorings are defined according to~\eqref{y-any-char}. Multiplicities~$a_k$ are as follows:
\begin{equation}\label{ak}
a_1=1,\quad a_2=1,\quad a_3=1,\quad a_4=|F|,\quad a_5=|F|^4,
\end{equation}
where $|F|$ is the cardinality of~$F$.
\end{theorem}

\begin{proof}
Let $C\subset \partial\Delta^5$ and~$\bar C$ be clusters of pentachora as in Subsection~\ref{ss:full}, i.e., the l.h.s.\ and r.h.s.\ of a Pachner move. Combining Lemma~\ref{l:d:R} and~\eqref{tnb}, we see that the permitted colorings of the common boundary $\partial C = \partial \bar C$ are generated, according to~\eqref{d:psi}, by the edges \emph{belonging to this boundary}. This means the item~\ref{i:CbarC} in Subsection~\ref{ss:full} is fulfilled.

As for the multiplicities~\eqref{ak}, they are checked by a direct calculation.
\end{proof}

\begin{remark}
Multiplicities~\eqref{ak} arise of course because there is a one-dimen\-sional space of colorings of \emph{only inner tetrahedra} (that is, the boundary tetrahedra are colored with zeros) if $C$ contains four pentachora, and a four-dimen\-sional space of colorings of only inner tetrahedra if $C$ contains five pentachora. The obvious reason why these colorings of only inner tetrahedra occur is that they are generated, according to~\eqref{psi}, by (nonzero) colorings of only \emph{inner edges}. There is one inner edge if $C$ contains four pentachora, and five inner edges if $C$ contains five pentachora. In the latter case, however, the five colorings are linearly dependent: it can be checked that \emph{their sum is the identical zero}, compare formula~\eqref{o:g} and Subsection~\ref{ss:mr} in Appendix~\ref{s:app}.
\end{remark}

\subsection{`Rough' invariant from the total number of colorings}

It follows from Theorem~\ref{th:g} that we can construct a PL manifold invariant from the total number of colorings as follows.

Recall (see the end of Subsection~\ref{ss:perm}) that~$R_C$ means the set of permitted colorings of a simplicial complex~$C$, and vertical bars mean the cardinality of a set. In the linear case, $R_C$ is obviously a linear space over~$F$.

\begin{theorem}\label{th:r}
The following quantity:
\begin{equation}\label{rough}
\frac{|R_M|}{|F|^{2N_0+N_4/2}},
\end{equation}
or its base~$|F|$ logarithm, which we call~$I_{\mathrm{rough}}$:
\begin{equation}\label{logrough}
I_{\mathrm{rough}}(M) = \dim_F R_M - 2N_0 - \frac{1}{2}N_4,
\end{equation}
is a PL manifold invariant, for a given finite field~$F$. Additionally, $I_{\mathrm{rough}}(M)$ depends only on the characteristic of~$F$.
\end{theorem}

\begin{proof}
Indeed, it is not hard to see using~\eqref{ak}, and counting how many vertices and pentachora appear or disappear in a Pachner move, that \eqref{rough} does not change under any move. As for the second statement of the theorem, note that $\dim_F R_M$ is obtained from a system of \emph{linear} equations, so it remains the same under any field extension.

What remains is to show that \eqref{rough} does not depend on the \emph{order} of vertices. Indeed, here is how we can change the position of any one vertex~$v$ in this ordering. We do any chain of Pachner moves that removes~$v$ from the triangulation; this removal is of course performed by a move 5--1. Then we do all this chain backwards, but when doing the corresponding move 1--5, we \emph{change the position} of~$v$ in the order of vertices into any other we like. Such possibility is of course ensured by the fact that we have a \emph{full} hexagon.
\end{proof}

\section{Polynomial hexagon cohomologies}\label{s:cohomology}

\subsection{Polynomial hexagon chain complex}\label{ss:c}

Like in the case of quandles~\cite[Chapter~4]{CKS}, a \emph{cohomology theory} can be proposed based on our set theoretic hexagon relations. There are actually many versions of cohomologies; here we will be using the \emph{polynomial cohomologies} defined as explained below, and taking advantage of the \emph{linearity} of the conditions~\eqref{y-any-char} defining the permitted colorings.

Our polynomials will be \emph{homogeneous} of a fixed degree~$\kappa$, and over a \emph{prime field}~$\mathbb F_p$. We will see that the case $\kappa=3$, \ $p=2$ is already very interesting.

Our letters $x_t$ and~$y_t$ will mean, in this Section and in the next Section~\ref{s:actions}, simply \emph{variables} over~$\mathbb F_p$, corresponding to tetrahedron~$t$. That is, each time we speak about an algebraic construction related to a simplicial complex, we imply that everything happens in the \emph{ring of polynomials} of variables $x_t,y_t$, where $t$ runs over all tetrahedra in that complex. Only in Section~\ref{s:results} these variables will resume taking specific values in finite fields.

For each pentachoron~$u$, we consider the \emph{ideal\/}~$I_u$ generated by the five linear forms obtained as follows: for $u=12345$, subtract the r.h.s.\ of~\eqref{y-any-char} from the l.h.s., and take the entries of the resulting column; for any~$u$, make the substitution of vertices as described in Subsection~\ref{ss:o}.

\paragraph{Cochains:} $n$-cochains, $n\ge 3$, belong to the $n$-simplex~$\Delta^n$ with vertices $1,\ldots,n+1$. Let $\mathbf x_{\Delta^n}$ denote the set of all~$x_t$, \ $t\subset \Delta^n$, and  $\mathbf y_{\Delta^n}$ similarly the set of all~$y_t$. Also, let $I_{\Delta^n}$ denote the \emph{sum} of all ideals~$I_u$, \ $u\subset \Delta^n$.

By definition, $n$-cochains are those elements in the factor ring $\mathbb F_p[\mathbf x_{\Delta^n},\mathbf y_{\Delta^n}] / I_{\Delta^n}$ that arise from polynomials of degree~$\kappa$ in $\mathbb F_p[\mathbf x_{\Delta^n},\mathbf y_{\Delta^n}]$.

\begin{remark}
Simply speaking, $n$-cochains are polynomials of degree~$\kappa$ of any maximal linearly independent set of variables $x_t,y_t$. The dependencies are given here, of course, by the formulas of type~\eqref{y-any-char} for all pentachora~$u$.
\end{remark}

\paragraph{Coboundary operator:} let an $(n-1)$-cochain~$c$ be represented by a polynomial $P(\mathbf x_{\Delta^{n-1}},\mathbf y_{\Delta^{n-1}})$. Its coboundary $\delta^n c$ is, by definition, represented by the polynomial
\begin{equation}\label{cob}
\Bigl(\text{representative of }(\delta^n c)(\mathbf x_{\Delta^n},\mathbf y_{\Delta^n}) \Bigr)
= \sum_{i=1}^{n+1} (-1)^{(i-1)} P(\mathbf x_{\Delta_i^{n-1}},\mathbf y_{\Delta_i^{n-1}}).
\end{equation}

In~\eqref{cob}, $\Delta_i^{n-1}$ means the $(n-1)$-face lying \emph{opposite} vertex~$i$. This~$\Delta_i^{n-1}$ is identified with the standard~$\Delta^{n-1}$ (having vertices $1,\ldots,n$), together with all their $x_t$ and~$y_t$, in the spirit of Subsection~\ref{ss:o}. Namely, this identification follows from the identification of the vertices of~$\Delta_i^{n-1}$, taken in the increasing order of their numbers, with the vertices of~$\Delta^{n-1}$, taken also in the increasing order of their numbers.

The correctness of definition~\eqref{cob} can be easily checked.

The part of the cochain complex we will be dealing with here is as follows:
\begin{equation}\label{c}
C^3 \stackrel{\delta^4}{\longrightarrow} C^4 \stackrel{\delta^5}{\longrightarrow} C^5 .
\end{equation}
where $C^n$ is of course the group of $n$-cochains. More accurately, it should be written as $C^n=C_{p,\kappa}^n$, because, as we remember, it depends on the field characteristic~$p$ and the degree~$\kappa$ of the involved polynomials.

\subsection{Hexagon 4-cocycles, Pachner moves, and manifold invariants}\label{ss:f}

Here we explain the meaning of the fragment~\eqref{c}, and specifically of the group
\begin{equation}\label{H4}
H_{p,\kappa}^4=\Ker \delta^5 / \LinearImage \delta^4,
\end{equation}
for building invariants of four-dimen\-sional PL manifolds.

Operator~$\delta^5$ deals with the 5-simplex~$\Delta^5$. Recall (Subsection~\ref{ss:P}) that its boundary $\partial \Delta^5$ is the l.h.s.\ and r.h.s.\ together of any four-dimen\-sional Pachner move. If the l.h.s.\ of a Pachner move is part of a triangulation of an \emph{oriented} manifold~$M$, then there is a consistent orientation of all pentachora in this l.h.s., induced from~$M$; the same applies to the r.h.s.\ when it has replaced the l.h.s.

We now choose any finite field $F\supset \mathbb F_p$ of characteristic~$p$. Any polynomial 4-cochain~$c$, defined as in Subsection~\ref{ss:c}, determines then a function
\[
f_u\colon\quad (\text{permitted } F \text{-colorings of } u) \to F
\]
for each pentachoron~$u$ (`$F$-colorings' means of course that the colors take values in~$F$). We now define the `discrete action density'~$A_u$ (see Subsection~\ref{ss:phys} for the explanation of this terminology) as~$f_u$, taken with sign plus if the orientation of~$u$ in the manifold coincides with its orientation determined by the increasing order of its vertices, or taken with sign minus otherwise.

With such an agreement, the cocycle condition for~$c$ can be seen to imply
\begin{equation}\label{cocycle}
\sum_{u\subset \text{l.h.s.}} A_u = \sum_{u\subset \text{r.h.s.}} A_u ,
\end{equation}
where the l.h.s.\ and r.h.s.\ are those of the Pachner move.

\begin{theorem}\label{th:gen}
 \begin{enumerate}\itemsep 0pt
  \item\label{i:cocycle} Let $M$ be a triangulated, orientable, closed, four-dimen\-sional PL manifold. Let $c$ be a polynomial hexagon 4-cocycle, and~$A_u$ constructed as explained above. Then, in addition to~\eqref{logrough}, there are the following invariants:
\begin{equation}\label{inv-g}
\mathrm P_M(v)=\frac{\# v}{|R_M|},
\end{equation}
where $\# v$ is the multiplicity of value $v\in F$---the number of times the quantity
\begin{equation}\label{action}
\mathcal A = \sum_{u\subset M} A_u
\end{equation}
takes the value~$v$ when the coloring runs over all permitted colorings. The ordering of triangulation vertices used for calculating invariants~\eqref{inv-g} can be arbitrary.
  \item\label{i:coboundary} Invariants~\eqref{inv-g} depend only on the cohomology class (see~\eqref{H4}) of~$c$.
 \end{enumerate}
\end{theorem}

Recall that $|R_M|$ is the number of all permitted $F$-colorings; $\mathrm P_M(v)$ may thus be called the \emph{probability} of value~$v$.

\begin{proof}
 \begin{itemize}\itemsep 0pt
  \item[\ref{i:cocycle}] Indeed, our cocycle property~\eqref{cocycle} and the full hexagon for~$R$~\eqref{y-any-char} ensure together that both $\# V$ and~$|R_M|$ can only acquire the same multiplier during a Pachner move. Independence of a vertex ordering is proved in the same way as in Theorem~\ref{th:r}.
  \item[\ref{i:coboundary}] Let $c$ be a coboundary of a cochain represented by a polynomial $P(\mathbf x_{\Delta^{n-1}}, \mathbf y_{\Delta^{n-1}})$, see~\eqref{cob}. Choose an arbitrary $(n-1)$-simplex~$\Delta_i^{n-1}$; it belongs to two pentachora. Then trace the signs with which the two corresponding terms $P(\mathbf x_{\Delta_i^{n-1}}, \mathbf y_{\Delta_i^{n-1}})$ enter in the action~\eqref{action}, using \eqref{cob} and the definition of~$A_u$ in this Subsection, and check that these signs are opposite.
 \end{itemize}
\end{proof}
    
\subsection{Hexagon 4-cocycles and manifold invariants in physical terms}\label{ss:phys}

In physical terms, the quantity~$\mathcal A$ in~\eqref{action} can be called (an analogue of) \emph{action}, whilst individual summands~$A_u$ represent a discrete analogue of action (or \emph{Lagrangian}) \emph{density}. Invariants~\eqref{inv-g} can be organized into a \emph{state-sum} form using homomorphisms
\begin{equation}\label{h}
\mathrm e\colon\;G\to \mathbb C^*
\end{equation}
of our abelian group into the \emph{multiplicative} group of complex numbers; any homomorphism~\eqref{h} can play the role of the exponential function in a `Feynman integral over all (physical) fields', which turns, in our case, into the sum
\[
\sum_{\substack{\text{all permitted}\\  \text{colorings}}}\mathrm e(\mathcal A).
\]

This is a standard idea (compare~\cite[Subsection~4.3.1]{CKS}) that we will not delve into. Instead, we just remark that our invariants can thus be given a form of a `discrete TQFT'---topological quantum field theory on piecewise linear four-manifolds.

\subsection{Two simple but useful lemmas}\label{ss:2lem}

\begin{lemma}\label{l:s4}
If a triangulated manifold~$M$ is PL homeomorphic to the sphere~$S^4$, and $c$ is a hexagon 4-cocycle, then action~\eqref{action} vanishes:
\[
\mathcal A=0\qquad \text{for any permitted coloring}.
\]
\end{lemma}

\begin{proof}
Indeed, Pachner theorem says that the triangulation of~$M$ can be transformed into one isomorphic to~$\partial \Delta^5$ by a sequence of Pachner moves, and $\mathcal A$ does not change under any move, due to the cocycle condition~\eqref{cocycle}. Then, the same condition asserts that $\mathcal A=0$ for~$\partial \Delta^5$.
\end{proof}

\begin{lemma}\label{l:b4}
If a triangulated manifold~$M$ with boundary~$\partial M$ is PL homeomorphic to the ball~$B^4$, and $c$ is a hexagon 4-cocycle, then the value of action~\eqref{action}, for a given permitted $F$-coloring of~$M$, is determined uniquely by the coloring of the \emph{boundary}~$\partial M$.
\end{lemma}

\begin{proof}
Let $\bar M$ be a copy of~$M$, but with the opposite orientation. Glue $\bar M$ to~$M$ by identifying $\partial M$ and~$\partial \bar M$ in the natural way, so that $M$ and~$\bar M$ form together a (manifold PL homeomorphic to) sphere~$S^4$. According to Lemma~\ref{l:s4}, $\mathcal A=\mathcal A_M + \mathcal A_{\bar M}=0$ for the resulting manifold, so $\mathcal A_M$ clearly cannot change if only the coloring of \emph{inner} tetrahedra in~$M$ is changed.
\end{proof}

\section[Action densities of degrees $\kappa \le 6$ in characteristics $p=2$, $3$ and~$5$]{Action densities of degrees $\boldsymbol{\kappa \le 6}$ in characteristics $\boldsymbol{p=2}$, $\boldsymbol{3}$ and~$\boldsymbol{5}$}\label{s:actions}

We present here the results of computer calculations---explicit formulas for 4-cocycles representing \emph{bases} of all cohomology spaces~$H_{p,\kappa}^4$~\eqref{H4} for $p=2$, $3$ and~$5$, and $\kappa\le 6$. The notations we will be using are as follows. Let $u=ijklm$, \ $i<j<k<l<m$, be a triangulation pentachoron whose orientation induced from the manifold coincides with that determined by the order of its vertices. Let $x_t$ be the half of variables on the 3-faces of~$u$ (the other half is determined, as we remember, from the relations of type~\eqref{y-any-char}). We denote below, for the ease of perception,
\begin{equation*}
a=x_{jklm},\quad b=x_{iklm},\quad c=x_{ijlm},\quad d=x_{ijkm},\quad e=x_{ijkl}
\end{equation*}
(compare the order of variables in~\eqref{y-any-char}).

\subsection[Characteristic $p=2$]{Characteristic $\boldsymbol{p=2}$}\label{ss:p=2}

One important note about $p=2$ is that the orientations play actually no role here. Whatever was said above about the orientations, can be ignored for $p=2$; in particular, the arising invariants work also for \emph{non-orientable} manifolds~$M$, in contrast with other characteristics~$p$.

\subsubsection*{Degree $\boldsymbol{\kappa=1}$}

There are no nontrivial cocycles for this~$\kappa$.

\subsubsection*{Degree $\boldsymbol{\kappa=2}$}

The cohomology space $H_{2,2}^4$ is one-dimen\-sional, and its basis is represented by the following cocycle:
\begin{align}
c_1 = de+ce+ae+cd+bd+c^2+bc+ac+ab. \label{2,2}
\end{align}

\subsubsection*{Degree $\boldsymbol{\kappa=3}$}

The cohomology space $H_{2,3}^4$ is two-dimen\-sional, and its basis is represented by the following cocycles:
\begin{align}
c_1 & = bde+bce+ace+acd+abd, \label{2,3a} \\
c_2 & = de^2+ce^2+ae^2+c^2d+b^2d+c^3+ac^2+b^2c+ab^2. \label{2,3b}
\end{align}

\subsubsection*{Degree $\boldsymbol{\kappa=4}$}

The cohomology space $H_{2,4}^4$ is three-dimen\-sional, and its basis is represented by the following cocycles:
\begin{align*}
c_1 & = ce^3+be^3+bde^2+ace^2+b^2e^2+c^3e+b^2ce+ac^2d+ab^2d, \\
c_2 & = bce^2+b^2e^2+b^2de+ac^2e+b^3e+bd^3+ad^3+abd^2+a^2d^2+a^2cd+b^3d, \\
c_3 & = d^2e^2+c^2e^2+a^2e^2+c^2d^2+b^2d^2+c^4+b^2c^2+a^2c^2+a^2b^2.
\end{align*}
Note that $c_3$ here is obtained by applying the Frobenius endomorphism to the cocycle~\eqref{2,2}. 

\subsubsection*{Degree $\boldsymbol{\kappa=5}$}

The cohomology space $H_{2,5}^4$ is again three-dimen\-sional, and its basis is represented by the following cocycles:
\begin{align*}
c_1 & = de^4+ce^4+ae^4+c^4d+b^4d+c^5+ac^4+b^4c+ab^4, \\
c_2 & = ce^4+be^4+c^2e^3+b^2e^3+b^2de^2+c^3e^2+ac^2e^2+b^2ce^2+b^3e^2+c^4e+b^4e \\
 & +bd^4+ad^4+b^2d^3+a^2d^3+b^3d^2+ab^2d^2+a^2c^2d+b^4d, \\
c_3 & = bd^2e^2+b^2de^2+bc^2e^2+ac^2e^2+b^2ce^2+a^2ce^2+b^2d^2e+b^2c^2e+a^2c^2e \\
 & +ac^2d^2+a^2cd^2+ab^2d^2+a^2bd^2+a^2c^2d+a^2b^2d.
\end{align*}

\subsubsection*{Degree $\boldsymbol{\kappa=6}$}

The cohomology space $H_{2,6}^4$ is four-dimen\-sional, and its basis is represented by the following cocycles:
\begin{align*}
c_1 & = ce^5+be^5+bde^4+c^2e^4+ace^4+c^4e^2+b^4e^2+c^5e+b^4ce+ac^4d+ab^4d, \\
c_2 & = b^2d^2e^2+b^2c^2e^2+a^2c^2e^2+a^2c^2d^2+a^2b^2d^2, \\
c_3 & = d^2e^4+c^2e^4+a^2e^4+c^4d^2+b^4d^2+c^6+a^2c^4+b^4c^2+a^2b^4, \\
c_4 & = bce^4+b^2e^4+b^4de+ac^4e+b^5e+bd^5+ad^5+b^2d^4+abd^4+b^4d^2 \\
 & +a^4d^2+a^4cd+b^5d.
\end{align*}
Note that $c_2$ and~$c_3$ here are obtained by applying the Frobenius endomorphism to the cocycles \eqref{2,3a} and~\eqref{2,3b}. 

\subsection[Characteristic $p=3$]{Characteristic $\boldsymbol{p=3}$}

\subsubsection*{Degrees $\boldsymbol{\kappa=1}$ and~$\boldsymbol{3}$}

There are no nontrivial cocycles for these~$\kappa$.

\subsubsection*{Degree $\boldsymbol{\kappa=2}$}

The cohomology space $H_{3,2}^4$ is one-dimen\-sional, and its basis is represented by the following cocycle:
\begin{align}\label{3,2}
c_1 & = e^2+2d^2+2bd+ad+c^2+bc+2ac.
\end{align}

\subsubsection*{Degree $\boldsymbol{\kappa=4}$}

The cohomology space $H_{3,4}^4$ is two-dimen\-sional, and its basis is represented by the following cocycles:
\begin{align*}
c_1 & = e^4+de^3+ce^3+be^3+ae^3+2c^3e+b^3e+2c^3d\\
 & +b^3d+2c^4+2bc^3+2ac^3+b^3c+b^4+ab^3, \\
c_2 & = de^3+ce^3+2be^3+2ae^3+d^2e^2+bce^2+2b^2e^2+2a^2e^2+bd^2e+c^2de\\
 & +bcde+acde+2abde+bc^2e+b^2ce+2abce+b^3e+2ab^2e+2a^2be+2c^2d^2\\
 & +2bcd^2+2acd^2+2b^2d^2+abd^2+c^3d+2bc^2d+2b^2cd+2abcd+2a^2cd+b^3d\\
 & +ab^2d+a^2bd+2c^4+2bc^3+ac^3+2a^2c^2+b^3c+ab^2c+a^2bc+b^4+ab^3.
\end{align*}

\subsubsection*{Degree $\boldsymbol{\kappa=5}$}

The cohomology space $H_{3,5}^4$ is one-dimen\-sional, and its basis is represented by the following cocycle:
\begin{align*}
c_1 & = de^4+2ae^4+d^2e^3+2bde^3+c^2e^3+2b^2e^3+abe^3+2a^2e^3\\
 & +c^3de+2b^3de+c^4e+2bc^3e+2b^3ce+b^4e+2c^3d^2+b^3d^2\\
 & +c^4d+bc^3d+2ac^3d+2b^3cd+2b^4d+ab^3d+c^5+bc^4\\
 & +2ac^4+abc^3+a^2c^3+2b^3c^2+2b^4c+ab^3c+2ab^4+2a^2b^3.
\end{align*}

\subsubsection*{Degree $\boldsymbol{\kappa=6}$}

The cohomology space $H_{3,6}^4$ is two-dimen\-sional, and its basis is represented by the following cocycles:
\begin{align*}
c_1 & = e^6+2d^6+2b^3d^3+a^3d^3+c^6+b^3c^3+2a^3c^3, \\
c_2 & = c^2e^4+bce^4+b^2e^4+cd^2e^3+2bd^2e^3+2ad^2e^3+2c^2de^3+2bcde^3+2acde^3\\
 & +b^2de^3+abde^3+a^2de^3+bc^2e^3+2ac^2e^3+2b^2ce^3+abce^3+a^2ce^3+ab^2e^3\\
 & +2a^2be^3+d^4e^2+c^4e^2+bc^3e^2+b^3ce^2+2a^4e^2+2d^5e+bd^4e+2c^2d^3e\\
 & +2bcd^3e+2acd^3e+b^2d^3e+abd^3e+a^2d^3e+c^3d^2e+2b^3d^2e+2a^3d^2e+c^4de\\
 & +2bc^3de+2ac^3de+2b^3cde+2a^3cde+ab^3de+a^3bde+c^5e+bc^4e+2b^2c^3e\\
 & +abc^3e+a^2c^3e+b^3c^2e+2a^3c^2e+b^4ce+ab^3ce+a^3bce+2b^5e+2ab^4e\\
 & +2a^2b^3e+a^3b^2e+2a^4be+d^6+cd^5+2ad^5+2c^2d^4+2bcd^4+2acd^4+b^2d^4\\
 & +abd^4+a^2d^4+c^3d^3+bc^2d^3+b^2cd^3+abcd^3+2a^2cd^3+b^3d^3+2ab^2d^3+c^4d^2\\
 & +bc^3d^2+2ac^3d^2+b^3cd^2+2a^3cd^2+2b^4d^2+2ab^3d^2+2a^3bd^2+2a^4d^2+2bc^4d\\
 & +b^2c^3d+abc^3d+2a^2c^3d+b^3c^2d+2b^4cd+ab^3cd+a^3bcd+2a^4cd+b^5d\\
 & +ab^4d+2a^3b^2d+a^4bd+c^6+2bc^5+2ac^5+2ab^2c^3+a^3c^3+b^4c^2+ab^3c^2\\
 & +a^3bc^2+ab^4c+2a^3b^2c+a^4bc+a^2b^4+a^3b^3.
\end{align*}
Here $c_1$ is obtained by applying the Frobenius endomorphism to the cocycle~\eqref{3,2}.

\subsection[Characteristic $p=5$]{Characteristic $\boldsymbol{p=5}$}

\subsubsection*{Degrees $\boldsymbol{\kappa=1}$ and $\boldsymbol{3\le\kappa\le 5}$}

There are no nontrivial cocycles for these~$\kappa$.

\subsubsection*{Degree $\boldsymbol{\kappa=2}$}

The cohomology space $H_{5,2}^4$ is one-dimen\-sional, and its basis is represented by the following cocycle:
\begin{align}\label{5,2}
c_1 & = e^2+ce+4be+d^2+3bd+2ad+c^2+3bc+3ac+2b^2.
\end{align}

\subsubsection*{Degree $\boldsymbol{\kappa=6}$}

The cohomology space $H_{5,6}^4$ is two-dimen\-sional, and its basis is represented by the following cocycles:
\begin{align*}
c_1 & = e^6+2de^5+2ce^5+be^5+2ae^5+4c^5e+b^5e+3c^5d+2b^5d+3c^6+4bc^5\\
 & +3ac^5+2b^5c+b^6+2ab^5, \\
c_2 & = de^5+3ce^5+2be^5+4ae^5+c^2e^4+3bce^4+b^2e^4+2cd^2e^3+3bd^2e^3+ad^2e^3\\
 & +3c^2de^3+4bcde^3+3acde^3+4b^2de^3+2abde^3+4a^2de^3+3c^3e^3+2bc^2e^3\\
 & +ac^2e^3+3b^2ce^3+abce^3+2a^2ce^3+2ab^2e^3+3a^2be^3+a^3e^3+2cd^3e^2+3bd^3e^2\\
 & +ad^3e^2+c^2d^2e^2+bcd^2e^2+acd^2e^2+4b^2d^2e^2+3abd^2e^2+3c^3de^2+3bc^2de^2\\
 & +2ac^2de^2+3abcde^2+3a^2cde^2+b^3de^2+3a^2bde^2+4c^4e^2+4ac^3e^2+4a^2c^2e^2\\
 & +2b^3ce^2+ab^2ce^2+3a^2bce^2+3a^3ce^2+b^4e^2+4ab^3e^2+2a^2b^2e^2+a^3be^2\\
 & +3a^4e^2+d^5e+3cd^4e+2bd^4e+4ad^4e+2bcd^3e+4b^2d^3e+abd^3e+3a^2d^3e\\
 & +3c^3d^2e+2bc^2d^2e+2ac^2d^2e+b^2cd^2e+2abcd^2e+2b^3d^2e+3ab^2d^2e+4a^3d^2e\\
 & +2c^4de+4ac^3de+4b^2c^2de+2a^2c^2de+4b^3cde+4ab^2cde+4a^2bcde+2a^3cde\\
 & +3b^4de+2ab^3de+2a^2b^2de+3a^3bde+3a^4de+4c^5e+4bc^4e+2b^2c^3e\\
 & +4abc^3e+2a^2c^3e+4b^3c^2e+4ab^3ce+a^3bce+3b^5e+4a^2b^3e+3a^4be+4d^6\\
 & +4cd^5+4bd^5+2ad^5+2c^2d^4+2acd^4+4b^2d^4+2a^2d^4+3c^3d^3+bc^2d^3\\
 & +2ac^2d^3+2b^2cd^3+abcd^3+2a^2cd^3+3b^3d^3+ab^2d^3+2a^2bd^3+3a^3d^3+4c^4d^2\\
 & +3bc^3d^2+3ac^3d^2+b^2c^2d^2+2abc^2d^2+2b^3cd^2+ab^2cd^2+3a^3cd^2+b^4d^2\\
 & +ab^3d^2+4a^3bd^2+c^5d+bc^4d+3b^2c^3d+2abc^3d+a^2c^3d+b^3c^2d+2ab^2c^2d\\
 & +3a^2bc^2d+a^3c^2d+ab^3cd+3a^2b^2cd+4a^3bcd+2a^4cd+2b^5d+a^2b^3d\\
 & +2a^4bd+2c^6+4ac^5+ab^2c^3+a^2bc^3+3a^3c^3+3ab^3c^2+3a^2b^2c^2+2a^3bc^2\\
 & +2a^4c^2+3b^5c+2a^2b^3c+4a^4bc+4b^6+2ab^5.
\end{align*}

\subsection[Remark on characteristic $0$]{Remark on characteristic $\boldsymbol{0}$}

The same calculation as above in this Section can well be done also in characteristic~$0$. The only nontrivial hexagon cocycle discovered this way is in degree~$2$, namely
\begin{align}\label{0,2}
4e^2-6ce+6be-d^2+2bd-2ad+4c^2-8bc+2ac+3b^2
\end{align}
Note that cocycles \eqref{3,2} (characteristic~$p=3$) and \eqref{5,2} ($p=5$) are just reductions modulo~$p$ of~\eqref{0,2} (up to a possible sign), while \eqref{2,2} ($p=2$) is not!

\section{Calculation of invariants: one more important theorem, and first results for specific manifolds}\label{s:results}

Given the values of all variables~$x_t$, all~$y_t$ are of course determined uniquely by~\eqref{y-any-char}. This allows us to identify the linear space~$R_M$ of all permitted colorings with the linear space $L\subset F^{N_3}$ given as follows. First, $F^{N_3}$ here consists of all $N_3$-tuples of (only) variables~$x_t$. Second, on each tetrahedron~$t$, there are \emph{two} values~$y_t$ given by~\eqref{y-any-char}, because we are considering a closed manifold~$M$ where $t$ belongs to two pentachora. By definition, the tuple of~$x_t$ is in~$L$ provided these two~$y_t$ coincide for each~$t$.

\subsection{A theorem that greatly simplifies the calculations}\label{ss:gs}

The multiplicities~$\# v$ in~\eqref{inv-g} turn out to be too huge even for a computer. Happily, our linear subspace~$L$ can be \emph{factored}: there exists a linear subspace $W\subset L$ such that the value~$v$ remains the same along any equivalence class in the factor~$L/W$. As $L/W$ is usually much smaller than~$L$, this makes our calculations much more feasible.

The subspace~$W$ is nothing but the space of all colorings generated by edges, according to \eqref{tnb} and~\eqref{d:psi}, and expressed in terms of only variables~$x_t$. That is, we take the first components of 2-columns~$\psi_{t,b}$. Denote these as~$\xi_{t,b}$, then, \eqref{tnb} and~\eqref{d:psi} say that
\[
\xi_{t,b}=0 \text{ \ if \ } b\not\subset t,
\]
and if $b$ does belong to~$t$, the numbers~$\xi_{t,b}$ are as follows:
\[
\begin{array}{c|cccccc} b & \mbox{\small 12} & \mbox{\small 13} & \mbox{\small 14} & \mbox{\small 23} & \mbox{\small 24} & \mbox{\small 34} \\ \hline \xi_{t,b} \strutik & -1 & 2 & -1 & -1 & 0 & 1 \end{array} \qquad \text{for \ } t=1234;
\]
for other tetrahedra~$t$, the necessary substitution must be made, in the style of Subsection~\ref{ss:o}.

\begin{theorem}\label{th:dual}
Let $F^{N_3}$ be, like in the beginning of this Section, the linear space of all sets of variables~$x_t\in F$. For each edge~$b$, consider the vector\/~$\Xi_b$ whose $t$-component is $\xi_{t,b}$, and let $W\subset F^{N_3}$ be the subspace spanned by all\/~$\Xi_b$. Then,
\begin{enumerate}\itemsep 0pt
 \item\label{i:WL} $W\subset L$. In other words, if we add any vector\/~$\Xi_b$ to a vector corresponding to a permitted coloring, we still get a permitted coloring,
 \item\label{i:LW} moreover, adding any\/~$\Xi_b$ to a permitted coloring does not change the action.
\end{enumerate}
\end{theorem}

\begin{proof}
\begin{itemize}\itemsep 0pt
 \item[\ref{i:WL}] This is just a rephrasing of Theorem~\ref{th:d:R}.
 \item[\ref{i:LW}] Take an arbitrary edge~$b$, and consider its \emph{star} within our given triangulation. Note that adding~$\Xi_b$ does not affect the coloring of the boundary of this star. It follows then from Lemma~\ref{l:b4} that the action stays indeed the same.
\end{itemize}
\end{proof}

\subsection[Calculation results for eight orientable and four unorientable manifolds using cocycle~\eqref{2,3a}]{Calculation results for eight orientable and four unorientable manifolds using cocycle~\textbf{(\ref{2,3a})}}\label{ss:fr}

We present here the manifold invariants arising from just one polynomial cocycle, namely~\eqref{2,3a}. We calculate the `rough' invariant~\eqref{logrough}, and then `probabilities'~\eqref{inv-g} for the four smallest fields of characteristic~2. As we said already in the beginning of Subsection~\ref{ss:p=2}, in characteristic~2 we can work with unorientable manifolds as well as orientable.

Below our notations are as usual: $S^n$ is an $n$-dimen\-sional sphere, $T^n$ is an $n$-dimen\-sional torus, $\mathbb RP^n$ is an $n$-dimen\-sional real projective space, $\mathbb CP^2$ is a \emph{complex} two-dimen\-sional projective space, and $S^2\tildetimes S^2$ denotes the \emph{twisted} product of two spheres~$S^2$.

As the reader can see, there are the same twelve manifolds in each table below in this Subsection, of which the first eight are orientable, and the remaining four unorientable.

\subsubsection*{`Rough' invariant~\textbf{(\ref{logrough})}---based on the total number of permitted colorings}
\[
\begin{array}{c|c}
M & I_{\text{rough}}(M) \\ \hline
S^4 \strutik & -6 \\ \hline
T^4 \strutik & 6 \\ \hline
S^2\times S^2 \strutik & -10 \\ \hline
S^2\tildetimes S^2 \strutik & -10
\end{array}
 \qquad
\begin{array}{c|c}
M & I_{\text{rough}}(M) \\ \hline
\mathbb CP^2 \strutik & -8 \\ \hline
S^2\times T^2 \strutik & 2 \\ \hline
S^3\times S^1 \strutik & 0 \\ \hline
\mathbb RP^3\times S^1 \strutik & 2
\end{array}
 \qquad
\begin{array}{c|c}
M & I_{\text{rough}}(M) \\ \hline
\mathbb RP^2\times S^2 \strutik & -4 \\ \hline
\mathbb RP^2\times T^2 \strutik & 4 \\ \hline
\mathbb RP^2\times \mathbb RP^2 \strutik & 0 \\ \hline
\mathbb RP^4 \strutik & -2
\end{array}
\]

\subsubsection*{`Refined' invariants~\textbf{(\ref{inv-g})}}
\begin{align*}
\text{For }F=\mathbb F_2\colon\qquad\qquad  &
       \begin{array}{c|c|c}
M             
& \mathrm P_M(0) & \mathrm P_M(1)  \\ \hline
S^4                \strutik           
&     1          &     0           \\ \hline
T^4  \strutik
&     1          &     0           \\ \hline
S^2\times S^2      \strutik 
&     1          &     0           \\ \hline
S^2\tildetimes S^2 \strutik
&    1/2         &    1/2          \\ \hline
\mathbb CP^2       \strutik 
&    1/2         &    1/2          \\ \hline
S^2\times T^2      \strutik 
&     1          &     0           \\ \hline
S^3\times S^1      \strutik 
&     1          &     0           \\ \hline
\mathbb RP^3\times S^1  \strutik 
&     1          &     0           \\ \hline
\mathbb RP^2\times S^2  \strutik 
&    3/4         &    1/4          \\ \hline
\mathbb RP^2\times T^2  \strutik 
&    9/16        &    7/16         \\ \hline
\mathbb RP^2\times \mathbb RP^2  \strutik 
&    1/2         &    1/2          \\ \hline
\mathbb RP^4  \strutik 
&    3/4         &    1/4    
       \end{array}
\end{align*}
\bigskip
\begin{align*}
\text{For }F=\mathbb F_4\colon\qquad\qquad  &
       \begin{array}{c|c|c|c}
M             
& \mathrm P_M(0) & \mathrm P_M(1) & \mathrm P_M(\text{any other value}) \\ \hline
S^4                \strutik           
&     1          &     0          &                     0               \\ \hline
T^4  \strutik
&   65/128       &   63/128       &                     0               \\ \hline
S^2\times S^2      \strutik 
&    5/8         &    3/8         &                     0               \\ \hline
S^2\tildetimes S^2 \strutik
&    5/8         &    3/8         &                     0               \\ \hline
\mathbb CP^2       \strutik 
&    1/4         &    3/4         &                     0               \\ \hline
S^2\times T^2      \strutik 
&    5/8         &    3/8         &                     0               \\ \hline
S^3\times S^1      \strutik 
&     1          &     0          &                     0               \\ \hline
\mathbb RP^3\times S^1  \strutik 
&    5/8         &    3/8         &                     0               \\ \hline
\mathbb RP^2\times S^2  \strutik 
&    7/16        &    3/16        &                    3/16             \\ \hline
\mathbb RP^2\times T^2  \strutik 
&   67/256       &    63/256      &                   63/256            \\ \hline
\mathbb RP^2\times \mathbb RP^2  \strutik 
&    1/4         &    9/32        &                   15/64             \\ \hline
\mathbb RP^4  \strutik 
&    7/16        &    3/16        &                    3/16   
       \end{array}
\end{align*}
\bigskip
\begin{align*}
\text{For }F=\mathbb F_8\colon\qquad\qquad  &
       \begin{array}{c|c|c}
M             
& \mathrm P_M(0) & \mathrm P_M(\text{any other value}) \\ \hline
S^4 \strutik          
&       1        &                 0                   \\ \hline
T^4 \strutik          
&    71/512      &              63/512                 \\ \hline
S^2\times S^2 \strutik 
&     11/32      &                3/32                 \\ \hline
S^2\tildetimes S^2 \strutik 
&      1/8       &                 1/8                 \\ \hline
\mathbb CP^2       \strutik 
&      1/8       &                 1/8                 \\ \hline
S^2\times T^2      \strutik 
&     11/32      &                3/32                 \\ \hline
S^3\times S^1      \strutik 
&       1        &                 0                   \\ \hline
\mathbb RP^3\times S^1  \strutik 
&     11/32      &                3/32                 \\ \hline
\mathbb RP^2\times S^2  \strutik 
&     15/64      &                 7/64                \\ \hline
\mathbb RP^2\times T^2  \strutik 
&    519/4096    &              511/4096               \\ \hline
\mathbb RP^2\times \mathbb RP^2  \strutik 
&      1/8       &                 1/8                 \\ \hline
\mathbb RP^4  \strutik 
&     15/64      &                 7/64     
       \end{array}
\end{align*}
\bigskip
\begin{align*}
 \begin{matrix}\text{For }F=\mathbb F_{16};\\[\medskipamount]
               \text{here }\sqrt[5]{1}\text{ is}\\
               \text{any such }x\\
               \text{that }x^5=1\colon\end{matrix}\quad & 
       \begin{array}{c|c|c|c}
M             
& \mathrm P_M(0) & \mathrm P_M(\sqrt[5]{1}) & \mathrm P_M(\text{any other value}) \\ \hline
S^4                \strutik           
&     1          &     0          &                     0               \\ \hline
T^4  \strutik
&  2213/32768    &  1953/32768    &                2079/32768           \\ \hline
S^2\times S^2      \strutik 
&    23/128      &   3/128        &                   9/128             \\ \hline
S^2\tildetimes S^2 \strutik
&    23/128      &   3/128        &                   9/128             \\ \hline
\mathbb CP^2       \strutik 
&    1/16        &    3/16        &                     0               \\ \hline
S^2\times T^2      \strutik 
&    23/128      &   3/128        &                   9/128             \\ \hline
S^3\times S^1      \strutik 
&     1          &     0          &                     0               \\ \hline
\mathbb RP^3\times S^1  \strutik 
&    23/128      &   3/128        &                   9/128             \\ \hline
\mathbb RP^2\times S^2  \strutik 
&   31/256       &   15/256       &                   15/256            \\ \hline
\mathbb RP^2\times T^2  \strutik 
&  4111/65536    &  4095/65536    &                 4095/65536          \\ \hline
\mathbb RP^2\times \mathbb RP^2  \strutik 
&    1/16        &   129/2048     &                  255/4096           \\ \hline
\mathbb RP^4  \strutik 
&    31/256      &   15/256       &                   15/256  
       \end{array}
\end{align*}

\subsection{A few results for twisted tori}\label{ss:twt}

Below are our first results for (the simplest) twisted tori. First, we denote~$\tilde T_n^3$ the fiber bundle with base~$S^1$, fiber~$T^2$, and monodromy matrix~$\begin{pmatrix}1&n\\ 0&1\end{pmatrix}$. Such fiber bundles are \emph{three-dimen\-sional} twisted tori. Then we consider \emph{four-dimen\-sional} twisted tori~$\tilde T_n^4$ defined simply as direct products of~$\tilde T_n^3$ with a circle:
\[
\tilde T_n^4 = \tilde T_n^3 \times S^1.
\]
We also compare the results with those for the usual torus $\tilde T_0^4 = T^4$.

This time, we make calculations using three different cocycles: \eqref{2,3a}, \eqref{2,3b}, and their sum.

\begin{remark}
And we will see that the results for the sum of two cocycles do \emph{not} appear to be easily predicted from the results for the summands.
\end{remark}

\subsubsection*{`Rough' invariant~\textbf{(\ref{logrough})}}
\[
\begin{array}{c|c}
M & I_{\text{rough}}(M) \\ \hline
\tilde T_0^4=T^4 \strutik & 6 \\ \hline
\tilde T_1^4 \strutik & 4 \\ \hline
\tilde T_2^4 \strutik & 6 \\ \hline
\tilde T_3^4 \strutik & 4 \\ \hline
\tilde T_4^4 \strutik & 6
\end{array}
\]

\subsubsection*{`Refined' invariants~\textbf{(\ref{inv-g})} using either cocycle~\textbf{(\ref{2,3a})} or cocycle~\textbf{(\ref{2,3b})}}

It turns out that the results are exactly the same for these cocycles. See, however, the results for their \emph{sum} below!

\begin{align*}
\text{For }F=\mathbb F_2\colon\qquad\qquad  &
       \begin{array}{c|c|c}
M             
& \mathrm P_M(0) & \mathrm P_M(1)  \\ \hline
\tilde T_0^4,\; \tilde T_1^4,\; \tilde T_2^4,\; \tilde T_3^4,\; \tilde T_4^4 	\strutik           
&     1          &     0   
       \end{array}
\end{align*}

\begin{align*}
\text{For }F=\mathbb F_4\colon\qquad\qquad  &
       \begin{array}{c|c|c|c}
M             
& \mathrm P_M(0) & \mathrm P_M(1) & \mathrm P_M(\text{any other value}) \\ \hline
\tilde T_0^4,\; \tilde T_2^4,\; \tilde T_4^4	\strutik           
&   65/128       &    63/128      &  0 \\ \hline
\tilde T_1^4,\; \tilde T_3^4	  \strutik
&    17/32       &     15/32      &  0  
       \end{array}
\end{align*}

\begin{align*}
\text{For }F=\mathbb F_8\colon\qquad\qquad  &
       \begin{array}{c|c|c}
M             
& \mathrm P_M(0) & \mathrm P_M(\text{any other value})  \\ \hline
\tilde T_0^4,\; \tilde T_2^4,\; \tilde T_4^4	\strutik           
&    71/512      &     63/512      \\ \hline
\tilde T_1^4,\; \tilde T_3^4	  \strutik
&    23/128      &     15/128     
       \end{array}
\end{align*}

\begin{align*}
\text{For }F=\mathbb F_{16}\colon\quad  &
       \begin{array}{c|c|c|c}
M             
& \mathrm P_M(0) & \mathrm P_M(\sqrt[5]{1}) & \mathrm P_M(\text{any other value}) \\ \hline
\tilde T_0^4,\; \tilde T_2^4,\; \tilde T_4^4	\strutik           
&   2213/32768   &    1953/32768   &  2079/32768   \\ \hline
\tilde T_1^4,\; \tilde T_3^4	  \strutik
&    173/2048    &    105/2048     &  135/2048     
       \end{array}
\end{align*}

\subsubsection*{`Refined' invariants~\textbf{(\ref{inv-g})} using cocycle \textbf{(\ref{2,3a})}+\textbf{(\ref{2,3b})}}
\begin{align*}
\text{For }F=\mathbb F_2\colon\qquad\qquad  &
       \begin{array}{c|c|c}
M             
& \mathrm P_M(0) & \mathrm P_M(1)  \\ \hline
\tilde T_0^4,\; \tilde T_1^4,\; \tilde T_2^4,\; \tilde T_3^4,\; \tilde T_4^4 	\strutik           
&     1          &     0   
       \end{array}
\end{align*}

\begin{align*}
\text{For }F=\mathbb F_4\colon\qquad\qquad  &
       \begin{array}{c|c|c|c}
M             
& \mathrm P_M(0) & \mathrm P_M(1) & \mathrm P_M(\text{any other value}) \\ \hline
\tilde T_0^4,\; \tilde T_4^4	\strutik           
&     1          &      0         &  0 \\ \hline
\tilde T_2^4     \strutik 
&     5/8        &     3/8        &  0  \\ \hline
\tilde T_1^4,\; \tilde T_3^4	 \strutik
&     1          &      0         &  0 
       \end{array}
\end{align*}

\begin{align*}
\text{For }F=\mathbb F_8\colon\qquad\qquad  &
       \begin{array}{c|c|c}
M             
& \mathrm P_M(0) & \mathrm P_M(\text{any other value})  \\ \hline
\tilde T_0^4,\; \tilde T_4^4	\strutik           
&     1          &     0           \\ \hline
\tilde T_2^4     \strutik 
&    11/32       &    3/32         \\ \hline
\tilde T_1^4,\; \tilde T_3^4      \strutik 
&     1          &     0     
       \end{array}
\end{align*}

\begin{align*}
\text{For }F=\mathbb F_{16}\colon\qquad  &
       \begin{array}{c|c|c|c}
M             
& \mathrm P_M(0) & \mathrm P_M(\sqrt[5]{1}) & \mathrm P_M(\text{any other value}) \\ \hline
\tilde T_0^4,\; \tilde T_4^4	\strutik           
&     1          &      0         &  0      \\ \hline
\tilde T_2^4     \strutik 
&   23/128       &     3/128      &  9/128  \\ \hline
\tilde T_1^4,\; \tilde T_3^4      \strutik 
&     1          &      0         &  0  
       \end{array}
\end{align*}

\section{Discussion}\label{s:d}  

Here are some possible directions for further research.  
  
\subsubsection*{Nonconstant set-theoretic relations}

The invariants proposed in this paper correspond to the simplest case of `constant'---essentially, the same for all pentachora~$u$---set~$R_u$ of permitted colorings. More general possibilities are known, however, where $R_u$ depends on a `usual' simplicial 2-cocycle, or, essentially, its cohomology class $h\in H^2(M,F^*)$, see~\cite{bosonic}. This means that one deformation of our invariant~$I_{\mathrm{rough}}(M)$ can be done even before quantization!

\subsubsection*{Nonconstant hexagon cocycles}

Also, our hexagon cocycles~$A_u$ in Section~\ref{s:actions} are the same for all pentachora~$u$. Some calculations suggest, however, that $A_u$ can be variable, even for a constant~$R_u$. How exactly $A_u$ can vary for different pentachora, remains at the moment a complete mystery.

\subsubsection*{Further calculations and a probabilistic approach}

There are more complicated hexagon cocycles than~\eqref{2,3a}, and of course there are much more 4-manifolds---see, e.g., the book~\cite{Scorpan}---for which we are going to make our calculations. In hard cases, it may make sense to calculate `probabilities'~\eqref{inv-g} \emph{approximately}, using many randomly chosen elements in~$L/W$ (see Subsection~\ref{ss:gs}), and compare these probabilities for different manifolds using statistical methods.

\subsubsection*{Exotic homologies}

Linear space~$R_M$ of permitted colorings in~\eqref{d:su} is obviously the kernel of a linear operator. Namely, subtract the r.h.s.\ of~\eqref{y-any-char} from its l.h.s., do this also for \emph{all} $N_3$ tetrahedra in the triangulation, and take the direct sum of the results (5-columns). This gives rise to the operator
\[
\hat R_M\colon \quad (\text{all triangulation colorings}) \to (5N_3\text{-columns}),
\]
and \eqref{d:su} acquires the form
\begin{equation}\label{uu}
\LinearImage \psi \subset \Ker \hat R_M.
\end{equation}
Inclusion~\eqref{uu} clearly makes think that linear mappings $\psi$ and~$\hat R_M$ may be part of a longer \emph{exotic chain complex}, while our factorspace~$L/W$ from Subsection~\ref{ss:gs} may be thought of as \emph{exotic homologies}.

\appendix

\section{Where the formulas for edge vectors and hexagon solutions come from}\label{s:app}

\emph{Edge vector} for edge~$b$, denoted~$\uppsi_b$, is by definition the image, under the mapping~$\psi$~\eqref{psi}, of such coloring of manifold~$M$ edges where $b$ has the color~$1\in \mathbb F$, while all the other edges have color~$0$.

This Appendix is, essentially, a brief exposition of one simple particular case of the results of paper~\cite{bosonic}.

\subsection{From linear dependencies between edge vectors to matrices: general position}\label{ss:rm}

Suppose there are linear dependencies between vectors~$\uppsi_b$ corresponding to edges~$b$ coming from each given vertex~$i$:
\begin{equation}\label{o:g}
\sum_j \gamma_{ij} \uppsi_{ij} = 0,
\end{equation}
where the sum is taken over all vertices~$j$ joined by an edge with~$i$, and $\gamma_{ij}\in \mathbb F$ are some coefficients. The reason behind~\eqref{o:g} will be explained in Subsection~\ref{ss:mr}.

We consider in this Subsection the case of \emph{generic}~$\gamma_{ij}$. This means, in our case, that there are no other linear dependencies between the restrictions of~$\uppsi_b$ onto pentachora or/and tetrahedra than follow directly from~\eqref{o:g}. This further implies that $\mathbb F$ is big enough so that it can contain necessary coefficients~$\gamma_{ij}$ (compare the specific case in the end of Subsection~\ref{ss:ex} with non-generic~$\gamma_{ij}$).

No other restrictions are put on $\gamma_{ij}$. Note that, in particular, $\gamma_{ij}$ and~$\gamma_{ji}$ are not related in any way.

For each separate pentachoron~$u$, there are five (the number of vertices) linear dependencies between ten (the number of edges) restrictions of edge vectors. Consequently, these latter generate exactly a 5-dimensional space; this will be by definition the space~$R_u$ of permitted colorings of the pentachoron 3-faces.

\begin{theorem}\label{th:o:g}
Coefficients~$\gamma_{ij}$ determine subspace~$R_u \subset \mathbb F^{10}$ up to automorphisms of five two-dimen\-sional linear spaces~$V_t$ belonging to the 3-faces~$t\subset u$.
\end{theorem}

\begin{proof}
Take a tetrahedron $t\subset u$ and consider the six (the number of edges $b\subset t$) 2-columns~$\psi_{t,b}$. From~\eqref{o:g}, four (the number of vertices) linear dependencies follow between~$\psi_{t,b}$. This means that we are free to choose arbitrarily (linearly independent) columns for two of~$\psi_{t,b}$; another choice will correspond exactly to an automorphism of~$V_t$. The remaining four~$\psi_{t,b}$ are then determined uniquely.
\end{proof}

\subsection{From linear dependencies between edge vectors to matrices: specific examples}\label{ss:ex}

Choose now the coefficients~$\gamma_{ij}$ as follows:
\begin{equation}\label{o:gM}
\gamma_{ij} = \begin{cases} 1, & \text{if \ }i<j,\\ M, & \text{if \ }i>j, \end{cases}
\end{equation}
where $M$ is, at the moment, \emph{generic}. A \emph{direct calculation} (valid in any field characteristic) shows then that Theorem~\ref{th:o:g} remains in force, and the 2-columns~$\psi_{t,b}$ can be chosen as follows:
\begin{equation}\label{o:psi}
\begin{pmatrix} \psi_{t,ij} & \psi_{t,ik} & \psi_{t,il} & \psi_{t,jk} & \psi_{t,jl} & \psi_{t,kl} \end{pmatrix} = \begin{pmatrix} -1 & 1-M & M & M & 0 & -M \\ -1/M & 1/M & 0 & 0 & 1 & -1 \end{pmatrix} .
\end{equation}
Here tetrahedron $t=ijkl$, and the six $\psi_{t,b}$ are written out as a $2\times 6$ matrix, as we did in~\eqref{d:psi}. Subspace~$R_u$ consists, for an example pentachoron $u=12345$, of vectors satisfying
\begin{equation}\label{o:yx}
\begin{pmatrix} y_{2345}\\ y_{1345}\\ y_{1245}\\ y_{1235}\\ y_{1234} \end{pmatrix} = 
A \begin{pmatrix} x_{2345}\\ x_{1345}\\ x_{1245}\\ x_{1235}\\ x_{1234} \end{pmatrix}, 
\end{equation}
where
\begin{equation}\label{o:AM}
A = \begin{pmatrix}0 & (M-1)/M^2 & 1/M^2 & -1/M & (M-1)/M^2\\
0 & 1/M & 0 & -1/M & 1/M\\
1/M & (1-M)/M^2 & (M-1)/M^2 & 0 & 1/M^2\\
1/M & (1-2 M)/M^2 & (M-1)/M^2 & 1/M & (1-M)/M^2\\
0 & -1/M & 1/M & 0 & 0\end{pmatrix}.
\end{equation}

We think that the case $M=-1$ is especially interesting. The interest is due to the fact that the restrictions of linear dependencies~\eqref{o:g} onto either one pentachoron on one tetrahedron \emph{are themselves no longer independent}, and Theorem~\ref{th:o:g} is not valid for them. Nevertheless, $M=-1$ can well be substituted into \eqref{o:psi} and~\eqref{o:AM}, which was exactly how we obtained our formulas \eqref{d:psi} and~\eqref{y-any-char}.

\subsection{From matrices to linear dependencies between edge vectors}\label{ss:mr}

Now we consider a way backwards; we will do it only for one pentachoron. Let a five-dimen\-sional subspace $R_u\subset \mathbb F^{10}$ be given of permitted colorings of pentachoron $u=12345$; let this $R_u$ be generic. Here is how edge vectors are constructed, and how linear dependencies~\eqref{o:g} between them appear.

An edge vector has vanishing $x$- and $y$-components at the two tetrahedra \emph{not} not containing the given edge. This gives four conditions, and these single out exactly a one-dimen\-sional space of edge vectors from the five-dimen\-sional~$R_u$.

Consider now the edge vectors for edges $12$, $13$, $14$ and~$15$. They all have zero $x$- and $y$-components at tetrahedron~$2345$; note that the vanishing of these two components singles out a \emph{three-dimen\-sional} subspace from~$R_u$. As the four vectors lie in a three-dimen\-sional space, there is necessarily a linear dependence of type~\eqref{o:g} between them.

\end{document}